\documentclass[12pt]{elsarticle}

\usepackage{amsmath,amsfonts,amssymb,amsthm,yhmath}  
\usepackage{graphicx}
\usepackage{soul}
\usepackage{xcolor}

\def \k{\mathbf{k}}
\def \hx{\hat{x}}

\def \hs{\hat{s}}
\def \hu{\hat{u}}

\def \hsig{\hat{\sigma}}
\def \x{\mathbf{x}}
\def \y{\mathbf{y}}

\def \calG{\mathcal{G}}
\def \calH{\mathcal{H}}
\def \calO{\mathcal{O}}

\def \calT{\mathcal{T}}
\def \calV{\mathcal{V}}

\def \bbR{\mathbb{R}}

\def \bq{\begin{equation}}
\def \eq{\end{equation}}
\def \bqa{\begin{eqnarray}}
\def \eqa{\end{eqnarray}}

\newcommand{\dbar}[1]{\overline{\overline{#1}}}

\newtheorem{lem}{Lemma}
\newtheorem{prop}{Proposition}

\newtheorem{cor}{Corollary}

\begin{document}
\begin{frontmatter}

\title{A single layer representation of the scattered field for multiple scattering problems}

\author{Didier Felbacq, Anthony Gourdin and Emmanuel Rousseau}

\affiliation{organization={L2C, University of Montpellier, CNRS},%Department and Organization
	addressline={Place Bataillon}, 
	city={Montpellier},
	postcode={34095}, 
	country={France}}

\begin{abstract}
	\textcolor{black}{
The scattering of scalar waves by a set of scatterers is considered. It is proven that the scattered field can be represented as an integral supported by any smooth surface enclosing the scatterers. This is a generalization of the series expansion over spherical harmonics and spherical Bessel functions for spherical geometries. More precisely, given a set of scatterers, the field scattered by any subset can be expressed as an integral over any smooth surface enclosing the given subset alone. It is then possible to solve the multiple scattering problem by using this integral representation instead of an expansion over spherical harmonics. This result is used to develop an extension of the Fast Multipole Method in order to deal with subsets that are not enclosed within non-intersecting balls. 
}
\end{abstract}


\end{frontmatter}

\section{Introduction}
\textcolor{black}{
We consider the scattering of scalar waves by a finite set of obstacles in $\mathbb{R}^p, \, p=2,3$ in the harmonic regime, with a time dependence of $e^{-i\omega t}$. Let us denote $B(O,R)$ a ball of radius $R$ and center $O$ containing the scatterers. It is known, cf. \cite{martin}, that for $\x \in \bbR^p \setminus B(O,R)$ the scattered field can be represented as a series of spherical harmonics :
\bq\label{multi}
u^s(\x)=\left
\{\begin{array}{l}
\sum_{n,m} u^s_{nm} h^{(1)}_n(k x) Y_n^m(\hx) \hbox{ for } p=3\\
\sum_{n} u^s_{n} H^{(1)}_n(k x) e^{in\theta}  \hbox{ for } p=2
\end{array} \right.
,\, x>R.
\eq
Here, $h^{(1)}_n$ is the spherical Hankel function of first type and order $n$, $H^{(1)}_n$ is the Hankel function of first type and order $n$ cf. \cite{abram}, and $\theta$ is the polar angle of $\x$ in $\bbR^2$.
Whether the functions defined by these series can be extended inside the ball is a difficult problem known as Rayleigh hypothesis. It was essentially solved theoretically in the 80' for the 2D problem, cf. \cite{millar,cadilhac} where it was shown that the short answer is that, generically, it cannot be extended due to the presence of singularities. A numerical study for the 3D case is provided in \cite{auguie} but it seems that the question is largely open for that dimension.
More generally, if two sets of scatterers can be enclosed in two non-intersecting balls then the interaction between the two sets can be handled by using spherical harmonics expansions. This fails to be true when the two sets cannot be separated in this way.
}
%Let us denote by $B_i=B(O,R_i)$ a ball contained in $K$ and by $B_e=B(O,R_e)$ the smallest ball with center $O$ containing $\Omega$.
%Note that, if there is only one scatterer, i.e. $a$ is constant inside $K$, there is also a representation of the field inside $K$ by a series in the following form:
%\bq\label{multi2}
%u(\x)=\left
%\{\begin{array}{l}
%\sum_{n,m} u^s_{nm} j_n(k x) Y_n^m(\hx) \hbox{ for } p=3\\
%\sum_{n} u^s_{n} J_n(k x) e^{in\theta}\hbox{ for } p=2
%\end{array} \right.
%,\, x<R_i.
%\eq
%Here, $J_n$ (resp. $j_n)$ is the Bessel (resp. spherical Bessel) function of order $n$ \cite{abram}.
%When $\overline{\Omega}=K$ and the boundary $\Gamma$ is a sphere, both series can be matched on $\Gamma$ by choosing $O$ the center of the sphere. This leads to an explicit form of the scattering coefficients. By considering the traces of the field on the boundary, one can obtain a pseudo-differential operator relating the Fourier coefficients of the incident field to that of the scattered field. In the case where $\Gamma$ is not  a sphere and $K$ is a proper subset of $\Omega$,

\textcolor{black}{
The purpose of this work is to provide an extension of the series representation in the latter case, by means of integral representations of the scattered field. An early work in that direction can be found in \cite{maystre}, see also \cite[chap. 6.12]{martinbook}. More precisely, we show that, given a set of scatterers and an arbitrary smooth surface enclosing the scatterers, the scattered field can be represented by a single layer integral supported by the surface \cite{colton}. This allows to handle the situation where the scatterers cannot be enclosed inside non-intersecting balls, by replacing the spherical harmonics expansions by integral representations and to describe the interaction by means of these integrals. 
%This is in fact an extension of the celebrated Fast Multiple Method, in which the interaction between clusters of scatterers is handled by means of the spherical harmonic expansion. In our approach, the field scattered by a subset is represented by a single layer integral and the enclosing surface is non longer limited to be a sphere (see fig. XX).
%
%The point of this work is to provide a generalization of the series expansion, in the latter case. The generalization is obtained by showing that the field scattered by a set of scatterers can be represented as an integral over any smooth enough surface enclosing the scatterers. This integral involves a density defined over the surface and a kernel which is the Green function of the scattering problem.
As an example of this situation, let us consider the set of scatterers depicted in fig. \ref{astroids5}. This set of scatterers can be decomposed into subsets enclosed into regions $\Omega_j$ bounded by a surface $\Gamma_j$. If the domains $\Omega_j$ could be enclosed inside non-intersecting balls, we could use the Fast Multipole Method \cite{FMM} to account for the coupling between the subsets, by using a spherical harmonics expansion. This cannot be done here. We propose an extension of this algorithm by replacing the series expansions by the single layer integrals supported by the surfaces $\Gamma_j$. This allows to bypass the non-intersecting balls limitation of the Fast Multipole Method. Of course, the FFM still remains extremely relevant for large scattering configurations, see \cite{Bruno,Chiang,Ganesh,Gimbutas,Matsushima} for applications of the FMM.
\begin{figure}[h!]
	\begin{center}
		\includegraphics[width=8cm]{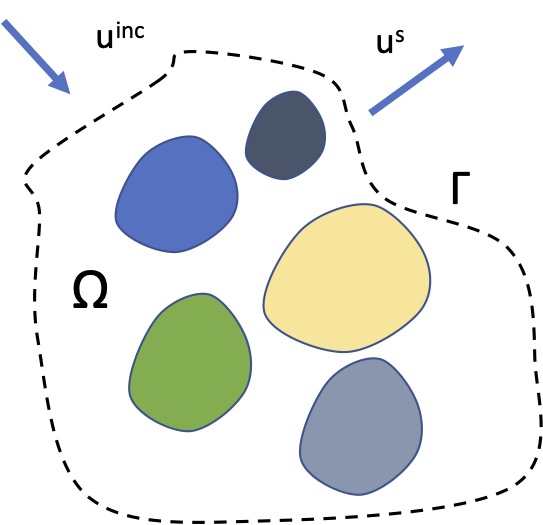}
		\caption{Sketch of the scattering problem under study. \label{scatprob}}
	\end{center}
\end{figure}
}
\section{Setting for the scattering problem}
Let us specify a few notations. The unit sphere of $\bbR^p$ is denoted $S^{p-1}$. For $\x \in \bbR^p$, we denote $x=\left | \x \right |$ the norm of $\x$, and  $\hx=\x/x$ , $k=\omega/c$. We denote $\calH_a u =\Delta u+k^2 a u$, where the potential $a$ belongs to $L^{\infty}(\bbR^p)$. The fundamental solution $g^+$ of the Helmholtz equation: $\calH_1 g^+=\delta_0$ with outgoing wave condition is: $g^+(\x)=-\frac{1}{4\pi x} e^{i k x}$ for $p=3$ and $g^+(\x)=-\frac{i}{4} H_0^{(1)}(kx)$ for $p=2$. The function $H_0^{(1)}$ is the Hankel function of first type \cite{abram}.

For $p=3$, the following expansion of the Green functions over the spherical harmonics holds, cf. \cite{FMM}:
\bq \label{greenexpanse}
g^+(\x-\x')=-\frac{e^{ik|\x-\x'|}}{4\pi |\x-\x'|}=-ik\sum_{n,m} j_n(k x_{<}) \, h_n^{(1)}(kx_{>})\, \overline{Y_n^m}(\hx') \, Y_n^m(\hx),
\eq
where $x_<=\min(x,x')$ and $x_>=\max(x,x')$.

%Note also that the scattered field can also be written
%$$
%u^s(r)=\int_{K} (\calH_{a,b}-\calH_{0}) (u)(r') g^+(r-r') dr'=\calT(u^i)
%$$
%which shows the existence of a densite $\sigma^s$ such that: $\calT(u^i)=\int_{\Omega} \sigma^s(r') g^+(r-r')dr'$.
%This is a decomposition in terms of incident and scattered fields. Another decomposition is possible, in terms of incoming and outgoing fields.

Let $\Omega$ be a bounded domain of $\mathbb{R}^p$ with boundary $\partial \Omega=\Gamma$.
 For $u \in H^1(\Omega)$ (the Sobolev space of function of $L^2(\Omega)$ with gradient in $L^2(\Omega)$, see \cite[chap. 2]{cessenat} for more results on Sobolev spaces), the interior traces \cite[chap. 2]{cessenat} of $u$  and its normal derivative on $\Gamma$ are denoted by:
\bq
\gamma^-(u)=\left .u \right |_{\Gamma},\, \gamma^-(\partial_n u)=\left .\partial_n u \right |_{\Gamma}.
\eq
For fields belonging to $H^1_{\rm loc}(\bbR^p \setminus \overline{\Omega})$, we denote the exterior traces by:
\bq
\gamma^+(u)=\left .u \right |_{\Gamma},\, \gamma^+(\partial_n u)=\left .\partial_n u \right |_{\Gamma}.
\eq
Given a field $u \in H^1_{\rm loc}(\bbR^p)$, we denote $[u]_{\Gamma}$ the jump of $u$ across $\Gamma$, i.e.:
\bq
[u]_{\Gamma}=\gamma^+(u)-\gamma^-(u) \hbox{ and } [\partial_n u]_{\Gamma}=\gamma^+(\partial_n u)-\gamma^-(\partial_n u).
\eq

%For the sake of self-containedness, we prove two existence results for the Helmholtz equation.
%
%Find $v$ satisfying $\calH_1 v=0$ in $\bbR^p \setminus B(O,R)$ with prescribed traces $\gamma_1(v)=\sum_{nm} v_{nm}Y_n^m(\hx) , \, \gamma_2(v)=\sum_{nm} w_{nm}Y_n^m(\hx)$.
%Write : $v(x)=\sum_{n,m} (a_{nm} j_n(kx)+b_{nm} h^{(1)}_n(kx)) Y_n^m(\hx)$.
%Then $$ a_{nm} j_n(kR)+b_{nm} h^{(1)}_n(kR)=v_{nm},\, a_{nm} j'_n(kR)+b_{nm} h^{(1)'}_n(kR)=w_{nm}/k$$
%and we obtain
%$$
%a_{nm}=\frac{w_{nm} h^{(1)}_n(kR) - k v_{nm} h^{(1)'}_n(kR)}{k(j'_n(kR)  h^{(1)}_n(kR)- h^{(1)'}_n(kR)j_n(kR) )},\,
% b_{nm}= -\frac{w_{nm} j_n(kR) - k j'_n(kR) v_{nm}}{k(j'_n(kR)  h^{(1)}_n(kR)- h^{(1)'}_n(kR)j_n(kR) )}.
%$$
Let us consider the time-harmonic scattering problem pictured in fig. \ref{scatprob}. The domain $\Omega$ contains a collection of scatterers characterized by a potential $a$ such that $a-1$ has a compact support $K \subset \Omega$. Let us remark that our analysis extend to the case of non-penetrable scatterers (i.e. no fields inside the scatterers).

The set of scatterers is illuminated by an incident field $u^{\rm inc}(x)$. The field  $u^{\rm inc}(x)$ satisfies the Helmholtz equation: $\calH_1 u^{\rm inc}=0$ in $\bbR^p$. 

%We assume further that $u^{\rm inc}(x)$ admits a representation in the form $u^{\rm inc}(x)=\int A(\hk) e^{ikx \hk.\hx}dk$, where $S^2 \ni \hk \to A(\hk)$ is a smooth function. This assumption is made in order to ensure that $u^{\rm inc}(x)$ admits an asymptotic expansion as $x \to \infty$, as discussed in \cite{martin}. We note that a plane wave $u(x)=e^{ikx}$, being associated with a distributional kernel $A(\hk)=\delta_{\hk}$, does not satisfy this hypothesis. However, the use of a mollifier $\varphi_{\e}$ \cite[chap.?,p.?]{schwartz} allows to take into account an approximate plane wave in the form $u^{\rm inc}(x)=\varphi_{\e} \star u(x)$ where $\star$ stands for the convolution product.

The scattering problem consists in finding the scattered field $u^s(\x)$ such that  the total field $u=u^{\rm inc}+u^s$ satisfies:
$$ \calH_a  u=0 \, ,$$ and $u^s$ satisfies a radiation condition at infinity: 
$$
\partial_n u^s -ik u^s=o\left(x^{-1}\right)  \hbox{ and } u^s(\x)=O\left(x^{-1}\right) \hbox{, when } x \to \infty.
$$

This scattering problem has a unique solution, as stated in the following lemma:
\begin{lem}\label{scatampl}
The scattered field $u^s$ exists and is unique. There is a linear operator $\calT$ relating $u^{\rm inc}$ to $u^s$: $u^s=\calT(u^{\rm inc})$.
\end{lem}

\begin{proof}
% $\calH_{a,b}(u^s+u^i)=0$ hence $\calH_{a,b}u^s=-\calH_{a,b}(u^i)$

$\calH_{a} (u^s)=(\calH_1-\calH_{a}) (u^{\rm inc})$ and $\calV\equiv\calH_1-\calH_{a} $ is null outside the compact region $K$. Then : $\calH_{1}(u^s)=\calV(u^s)+\calV(u^{\rm inc})$ and thus: ${u^s=(1-\calG_1 \calV)^{-1} \calG_1 \calV(u^{\rm inc})}$ where the inverse operator $\calG_1=\calH_1^{-1}$ is an integral convolution operator with kernel $g^+$.
\end{proof}
The existence of the resolvent operator is classical although rather subtle (see for instance \cite{cessenat,melrose}). 
This provides a decomposition of the total field in the form:
$$
u=u^{\rm inc}+\calT(u^{\rm inc}).
$$

In the following, the proofs of the results are given for $p=3$ and can be easily adapted for $p=2$ (or, in fact, any other dimension $>1$).

\textcolor{black}{
\section{The multiple scattering problem}
\label{intrep}
\subsection{Multiple scattering by a set of scatterers: integral representation of the fields}
We consider first the situation depicted in fig. \ref{scatprob} where all the scatterers are taken into account. Let us reformulate the scattering problem in a slightly different way.  Let $u^{\rm reg}$ satisfy the following problem : $$\calH_1 u^{\rm reg}=0 \hbox{ in } \Omega \hbox{ and } \gamma^-(u^{\rm reg})=\left . u^{\rm inc}\right|_{\Gamma}.$$ Then, apart from the discrete set of wavevectors that are the eigenvalues of $-\Delta$ with the homogeneous Dirichlet condition on $\Gamma$, it holds : $$u^{\rm reg}=u^{\rm inc} \hbox{ inside } \Omega.$$ Therefore, the scattered field can be seen as the response of the scatterers to the regular field $u^{\rm reg}$ given by: $u^s=\calT(u^{\rm reg})$. The last expression makes sense, since $\calT$ is an integral operator with compact support $K$ and hence it requires only that $u^{\rm reg}$ be defined in $\Omega$. If an incident field is given in $\bbR^p$, the regular field $u^{\rm reg}$ is the restriction of the incident field inside $\Omega$. If it is the regular field that is given inside $\Omega$, then it can be extended uniquely to $\bbR^p$ so as to satisfy $\calH_1 u^{\rm reg}=0$.
The regular field $u^{\rm reg}$ has an integral representation as described in the following result.
\begin{prop}
Let $u^{\rm reg}$ satisfy $\calH_1 u^{\rm reg}=0$ inside $\Omega$ and a given trace $\gamma^-(u^{\rm reg})$ on $\Gamma$. Then, there exists a unique density $\sigma^{\rm reg} \in H^{-1/2}(\Gamma)$ such that:
$$
u^{\rm reg}(\x)=\int_{\Gamma} \sigma^{\rm reg}(\x') g^+(\x-\x') ds(\x') ,\, \x \in \Omega.
$$ 
\end{prop}
\begin{proof}
Let us extend $u^{\rm reg}$ to $\bbR^3$ by imposing that it satisfies the following problem in $\bbR^3 \setminus \overline{\Omega}$:
\begin{eqnarray*}
\calH_1 u^{\rm reg}=0,\ \gamma^+(u^{\rm reg})=\gamma^-(u^{\rm reg}), \\
\partial_x u^{\rm reg}=ik u^{\rm reg} +o\left(\frac{1}{x}\right).
\end{eqnarray*}
The extended field is continuous through $\Gamma$ but its normal derivative is not. As a consequence, it holds:
$$
\calH_1 u^{\rm reg}=[\partial_n u^{\rm reg}]_{\Gamma} \delta_{\Gamma}.
$$
The proposition follows by defining $\sigma^{\rm reg}=[\partial_n u^{\rm reg}]_{\Gamma}$.
	The uniqueness follows from the following argument. Assume that there exists $\sigma^{\rm reg}$ such that: 
	$$
	\int_{\Gamma} \sigma^{\rm reg}(\x') g^+(\x-\x')ds(\x')=0,\, \forall \x \in \Omega.
	$$
	Define then $v(\x)=	\int_{\Gamma} \sigma^{\rm reg}(\x') g^+(\x-\x')ds(\x'),\, \forall \x \in \bbR^p$. The function $v$ is identically null inside $\Omega$. 
	Outside $\Omega$ it satisfies the following problem:
	$$
	\calH_1 v=0,\, \gamma^+(v)=0,
	$$
	and it satisfies a radiation condition. As a consequence, it is null by Rellich lemma, cf. \cite{cessenat}. Since $v$ satisfies $\calH_1 v=\sigma^{\rm reg} \delta_{\Gamma}$ in $\bbR^p$, we conclude that $\sigma^{\rm reg}=0$.	
\end{proof}
The purpose of this result is to stress that the incident field need only be defined locally, that is, inside $\Omega$. 
Our next result states that the scattered field can also be represented by an integral over $\Gamma$.
\begin{prop} \label{scatint}
Assume that $k^2$ is not an eigenvalue of $-\Delta$ inside $\Omega$ with the homogeneous Dirichlet condition on $\Gamma$. There exists a unique $\sigma^s \in H^{-1/2}(\Gamma)$ such that: 
\bq \label{integus}
u^s(\x)=\int_{\Gamma} \sigma^s(\x') g^+(\x-\x') ds(\x'), \x \in \bbR^p \setminus \Omega.
\eq
At infinity, the scattered field has the following asymptotic behavior:
$$
u^s(\x) =  \frac{-k}{4\pi}  \frac{e^{ikx}}{kx} \int_{\Gamma} \sigma^s(\x') e^{-ik \hx \cdot \x'}ds(\x')+\calO\left(\frac{1}{(kx)^2}\right).
$$
\end{prop}
\begin{proof}
Consider the field $\tilde{u}^s$ that is equal to $u^s$ outside $\Omega$ and that satisfies the following problem inside $\Omega$:
$$
\calH_1 \tilde{u}^s=0 , \, \gamma^-(\tilde{u}^s)=\gamma^+( u^s ).
$$
Over $\bbR^p$, it satisfies: $\calH_1 \tilde{u}^s=[\partial_n \tilde{u}^s]_{\Gamma} \delta_{\Gamma}.$
Given the outgoing wave condition at infinity, this gives:
$$
u^s(\x)=\int_{\Gamma} [\partial_n \tilde{u}^s]_{\Gamma} \, g^+(\x-\x') ds(\x'),
$$
and consequently the existence of the density $\sigma^s=[\partial_n \tilde{u}^s]_{\Gamma} $ belonging to $H^{-1/2}(\Gamma)$.
The asymptotic form is obtained by using the asymptotic form of the Green function:
\bq \label{green}
g^{\pm}(\x-\x') =_{x \to \infty}   -\frac{e^{\pm i k x}}{4\pi x} e^{\mp ik \x\cdot\x'}+\calO\left(\frac{x'^2}{x^2}\right),
\eq 
The unicity proceeds from the following argument. Assume that there is a $\sigma^s$ such that:$$\int_{\Gamma} [\partial_n \tilde{u}^s]_{\Gamma} g^+(\x-\x') ds(\x')=0,\forall \x \in \bbR^p\setminus \Omega.$$ Then define: 
$$
v(\x)=\int_{\Gamma} \sigma^s(\x') g^+(\x-\x') d\x',\forall \x \in \bbR^p.
$$
The function $v$ satisfies:
$$
\calH_1 v=0, \hbox{ in } \Omega, \, \gamma^-(v)=0.
$$
Consequently, it is null iff $k^2$ is not an eigenvalue of $-\Delta$ with the homogeneous Dirichlet condition on $\Gamma$.
\end{proof}
%In order to have an integral representation of the fields, we state a result concerning the incident field. To do so, we first need a technical lemma. 
}
\textcolor{black}{
At this stage, one could wonder about the possibility of representing the incident field as an integral over $\Gamma$ that would be valid outside $\Omega$. For an incident field satisfying the Helmholtz equation in all $\bbR^p$, we cannot impose an outgoing wave condition and the kernel of the integral representation should be regular. This suggests that a natural integral representation should be in the following form:
\bq \label{uincint}
u^{\rm inc}(\x)=\int_{\Gamma} \sigma^{\rm inc}(\x') \Im(g^+(\x-\x')) ds(\x'), \, \x \in \bbR^p.
\eq
It turns out that this representation imposes strong constraints on the regularity of the field $u^{\rm inc}$. This is somewhat similar to the possibility of splitting a regular field defined over $\bbR^p$ into an incoming and an outgoing field. The possibility of this decomposition was analyzed in \cite{martin}, where it was shown that this was possible provided that the incident field were of the following form:
\bq \label{wavepack}
u^{\rm inc}(\x)=\int_{\bbR} A(\k) e^{i\k\cdot \x} d\k,
\eq
where $A(\k)$ is twice continuously differentiable.
If the field $u^{\rm inc}$ admits the integral representation (\ref{uincint}) over $\Gamma$, it can be decomposed as the sum of two pointwise converging series in the form:
\bq \label{2series}
	u^{\rm inc}(\x)=\sum_{nm} \frac{1}{2} i_{nm} h_n^{(1)}(kx)Y_n^m(\hx)+ \sum_{nm} \frac{1}{2} i_{nm} h_n^{(2)}(kx) Y_n^m(\hx),
\eq
as can be infered by noting that $\Im(g^+(\x-\x'))$ is the sum of two functions satisfying an outgoing wave condition and an incoming wave condition respectively. Therefore, it has also such a decomposition at infinity. However, the converse if false. That is, the fact that the field admits a decomposition at infinity does not imply that the series in (\ref{2series}) be (pointwise) convergent. In fact, the condition is much more stringent than the wavepacket representation in (\ref{wavepack}). Indeed, it is not fulfilled for a general wavepacket even with a very regular (e.g. analytic) spectrum $A(\k)$. Needless to say, it is false for a plane wave.
}
\textcolor{black}{
\subsection{The field scattered by a subset of the set of scatterers}
We consider the situation depicted in fig. \ref{subsetg}: a subset of the set of scatterers is contained in a bounded domain $\Omega_0$, with boundary $\Gamma^-$. Further more we consider an annular region $A$ whose boundaries are $\Gamma^-$ and $\Gamma^+$ (see fig. \ref{subsetg}). We denote $\Omega_1=\Omega_0 \cup A$.
From the point of view of $\Omega_0$, there are two incident fields: the "true" incident field $u^{\rm inc}$ and the "local" incident field, which is the field diffracted by the scatterers that are not contained in $\Omega_0$. Inside $A$, the total field $u$ can be decomposed into a local regular field $u^{\rm reg}_{\rm loc}$, which is the sum of the true incident field and the local incident field, and a local scattered field $u^s_{\rm loc}=u-u^{\rm reg}_{\rm loc}$, where $u$ is the total field.
These fields can be represented by integrals supported by $\Gamma^{\pm}$. 
}
\begin{figure}[h!]
	\begin{center}
		\includegraphics[width=8cm]{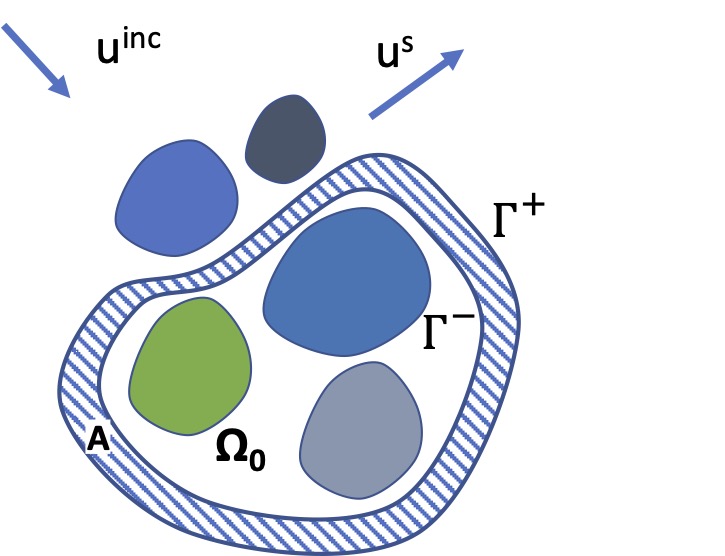}
		\caption{A subset of scatterers enclosed by $\Gamma^-$ is considered. \label{subsetg}}
	\end{center}
\end{figure}

\textcolor{black}{
\begin{prop}\label{densitygam}
	The total field $u$ can be written in the form
	\bq 
	u(\x)=u^s_{\rm loc}(\x)+u^{\rm reg}_{\rm loc}(\x),\, \x \in A,
	\eq
	%densities should be supported by the same cycle $\gamma_1$?
	where $u^s_{\rm loc}$ satisfies $\calH_1 u^s_{\rm loc}=0$ in $\bbR^p \setminus \Omega_0$ and a radiation condition and $u^{\rm reg}_{\rm loc}$ satisfies $\calH_1 u^{\rm reg}_{\rm loc}=0$ in $\Omega$.
	Moreover, there exist  $(\sigma^s,\sigma^{\rm reg})$ in $H^{-1/2}(\Gamma^-)\times H^{-1/2}(\Gamma^+)$ such that:
	\bq \label{decomp1}
	u^s_{\rm loc}(\x)=\int_{\Gamma^-} \sigma^s(\x') g^+(\x-\x')\, ds(\x'),
	\eq
	and
	\bq \label{decomp2}
	 u^{\rm reg}_{\rm loc}(\x)=\int_{\Gamma^+} \sigma^{\rm reg}(\x') g^+(\x-\x')\, ds(\x').
	\eq
\end{prop}
\begin{proof}
	Consider the restriction $\left. u \right|_A$ of $u$ to $A$ and extend $\left. u \right|_A$ by continuity across the boundaries $\Gamma^+$ and $\Gamma^-$, in such a way that the extended function $\tilde{u}$ satisfies: $$\Delta \tilde{u}+k^2 \tilde{u}=0 \hbox{ in }\Omega_0 \cup A \cup (\bbR^3 \setminus \bar{A}),$$
	with an outgoing wave condition at infinity. Then, function $\tilde{u}$ satisfies the following equation in the distributional meaning: $$(\Delta+k^2) \tilde{u}=\left[ \frac{\partial \tilde{u}}{\partial n}\right]_{\Gamma^+}\delta_{\Gamma^+}-\left[ \frac{\partial \tilde{u}}{\partial n}\right]_{\Gamma^-}\delta_{\Gamma^-},$$
	in $\bbR^p$. Therefore, since $\tilde{u}=u$ in $A$, the following representation of $u$ inside $A$ is obtained: $$u(\x)=\int_{\Gamma^+} \sigma^{\rm reg}(\x') g^+(\x-\x') ds(\x')+\int_{\Gamma^-} \sigma^s(\x') g^+(\x-\x') ds(\x'),$$ with ${\sigma^{\rm reg}= \left[ \frac{\partial u}{\partial n}\right]_{\Gamma^{+}}}$ and 
	${\sigma^{s}= -\left[ \frac{\partial u}{\partial n}\right]_{\Gamma^{-}}}$.
	The properties of $u^{s}_{\rm loc}$ and $u^{\rm reg}_{\rm loc}$ follow from potential theory.
	\end {proof}
}
%We can now deduce the following representation result for the total field:
%\begin{cor}\label{represent}
%The total field $u$ can be written in the form
%\bq \label{decomp}
%u(\x)=u^{\rm tot,+}(\x)+u^{\rm tot,-}(\x),
%\eq
%%densities should be supported by the same cycle $\gamma_1$?
%where 
%\bq
%u^{\rm tot,+}(\x)=\int_{\Gamma} \sigma^+(\x') g^+(\x-\x') d\x', \, u^{\rm tot,-}(\x)=\int_{\Gamma} \sigma^-(\x')  g^-(\x-\x') dx',
%\eq
%and $\sigma^+,\sigma^-$ belong to $H^{-1/2}(\Gamma)$. 
%\end{cor}
%
%\begin{proof}
%This is a direct consequence of theorem \ref{incid} and theorem \ref{scatint}. Putting the scattered field and the incident field together, we obtain:
%\bq
%\sigma^+=[\partial_n \tilde{u}^s]_{\Gamma}+\frac{1}{2} \sigma^{\rm inc},\, \sigma^-=-\frac{1}{2} \sigma^{\rm inc}.
%\eq
%\end{proof}
\textcolor{black}{
Since a decomposition of the field into a regular field and a scattered field was obtained, it is possible to use the scattering theory result exposed in lemma (\ref{scatampl}) to obtain the following:
\begin{cor}\label{fundlemma}
	There exists a linear operator $\calT_{\rm loc}$ such that:
	$u^s_{\rm loc}=\calT_{\rm loc}(u^{\rm reg}_{\rm loc}) $.
	This operator induces an operator $\calT_{\rm loc}^{\Gamma^-}$ supported by $\Gamma^-$ and such that:
	$$
	H^{1/2}(\Gamma^-)\ni\gamma( u^{\rm reg}_{\rm loc}) \xrightarrow{\calT_{\rm loc}^{\Gamma^-}} \sigma^s \in H^{-1/2}(\Gamma^-).
	$$
\end{cor}
\begin{proof}
	It remains to prove the existence of $\calT_{\rm loc}^{\Gamma^-}$. Given the relation $u^s_{\rm loc}=\calT_{\rm loc} u^{\rm reg}_{\rm loc}$, $\sigma^s$ is defined in proposition (\ref{densitygam}) and $\calT_{\rm loc}^{\Gamma^-}$ by the left invertibility of $\sigma^s \to u^s$ as defined by the integral in (\ref{integus})
\end{proof}
This shows that each subset can be characterized by an operator $\calT_{\rm loc}$ independently of the other scatterers. This is the basic result of the extension of the Fast Multipole Method that we propose. Instead of representing the field scattered by a subset as a series of spherical harmonics outside a sphere enclosing the scatterers, the scattered field is represented by an integral over some surface enclosing the subset. The interaction with the other subsets is then handled by using this integral representation directly, without resorting to a multipole expansion, as explained in detail in the following subsection.
\subsection{Extension of the Fast Multipole Method: The Fast Monopole Method}
It is important to remark that the single layer representation involves a surface $\Gamma$, enclosing the scatterers, that can be chosen at will. By this we mean that, given a set of scatterers and any smooth enough surface $\Gamma$ enclosing this set, the field scattered by this set  and the incident field can be represented by an integral over $\Gamma$. As already said, this result is in fact a generalization of the series expansion over spherical harmonics and spherical Bessel functions to an arbitrary surface. As a consequence, it is possible to split a given set of $N$ scatterers into several subsets containing $N_j$ scatterers, apply multiple scattering theory to each smaller subset, then use the representation by the single layers to couple the subsets in between them. The gain lies in the fact that the interaction between the $N$ scatterers is not handled by direct coupling, which results in $N^2$ operations, but instead through the boundary layers (\ref{decomp1}) which are discretized and represented in the form: $\sum_{p=1}^{P_j} \sigma_p^s H_0(k|\x-\x_p|)$, where the number $P_j$ is much smaller than $N_j$ (see the numerical applications for specific values, notably table \ref{tablevaleur}).
%\paragraph{\textbf{A proof of principle}}
We consider the example described in fig. \ref{subset}. It consists of a decomposition of the set in fig. \ref{scatprob} into two subsets $O^1$ and $O^2$.
\begin{figure}[h!]
	\begin{center}
		\includegraphics[width=8cm]{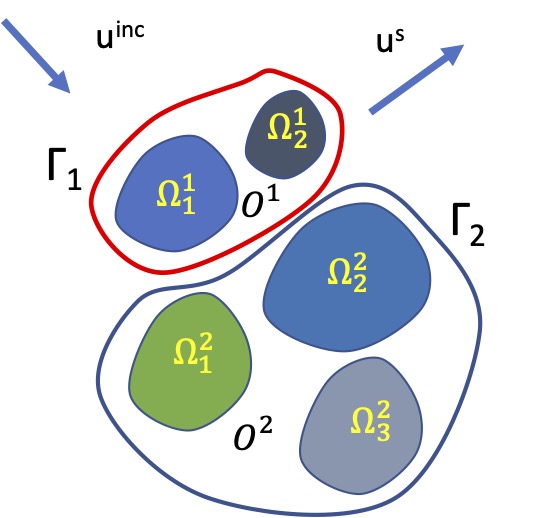}
		\caption{Sketch of the scattering problem with two subsets. \label{subset}}
	\end{center}
\end{figure}
The incident field is $u^{\rm inc}$. For each scatterer $\Omega^{1}_j \in O^{1}$, the local incident field is the sum of the "true" incident field $u^{\rm inc}(\x^1_j)$, the fields $u^{s,1,\rm loc}_k$ scattered by the other scatterers $\Omega^1_{k}, \, k \neq j$,  and the field coming from the other subset $O^2$ and given by the single layer representation:
\bq \label{sgllayer}
%u^{s,2}(\x^1_j)=\sum_{p=1}^{P_2} \sigma_p^{s,2} H_0^{(1)}(k|\x^1_j-\y_p|).
u^{s,2}(\x^1_j)=\int_{\Gamma_2} \sigma^{s,2}(\x') H_0^{(1)}(k|\x^1_j-\x'|)ds(\x').
\eq
Of course, there is the same set of relations obtained by making the exchange $1 \leftrightarrow 2$.
The local incident field is therefore:
\bq
u^{\rm inc,loc}(\x^1_j)=u^{\rm inc}(\x^{1}_j)+u^{s,2}(\x^1_j)+\sum_{k\neq j} u^{s,1,\rm loc}_k(\x^1_j).
\eq
Using corollary (\ref{fundlemma}), we can infer the existence of two sets of operators $\calT^1_j,\, \calT^2_j$ such that:
\bq
u^{s,\rm loc,1}_j(\x^1_j)=\calT^1_j (u^{\rm inc,loc}(\x^1_j)),\, u^{s,\rm loc,2}_j(\x^2_j)=\calT^2_j (u^{\rm inc,loc}(\x^2_j)) .
\eq
The system can now be closed by relating the local scattered fields $u^{s,\rm loc,1}$ and $u^{s,\rm loc,2}$ to $u^{s,1}$ and $u^{s,2}$ respectively and then to the densities $\sigma^1$ and $\sigma^2$. This is done first by summing the local scattered fields in each subset:
\bq
u^{s,1}(\x^1)=\sum_j u^{s,\rm loc,1}_j(\x^1_j),\, u^{s,2}(\x^2)=\sum_j u^{s,\rm loc,2}_j(\x^2_j),
\eq
then by solving the integral equation (\ref{decomp1}), which gives $\sigma^{s,1}$ and $\sigma^{s,2}$.
It should be noted that in order to put the algorithm into practice, an iterative algorithm is to be used, so as to avoid inverting a large matrix. Otherwise, the gain in memory would be lost, although the number of operations would still be reduced.
}
\section{Discussion and numerical examples}
\subsection{Representation of the field on the boundary of an inhomogeneous scatterer}

Let us consider the scattering of an electromagnetic plane wave in $E_{||}$ polarization by a collection of small cylinders contained in a domain whose cross section $\Omega$  is bounded by an curve $\Gamma$ (cf. fig.~\ref{astroid}). 
\begin{figure}[h!]
	\begin{center}
		\includegraphics[width=10cm]{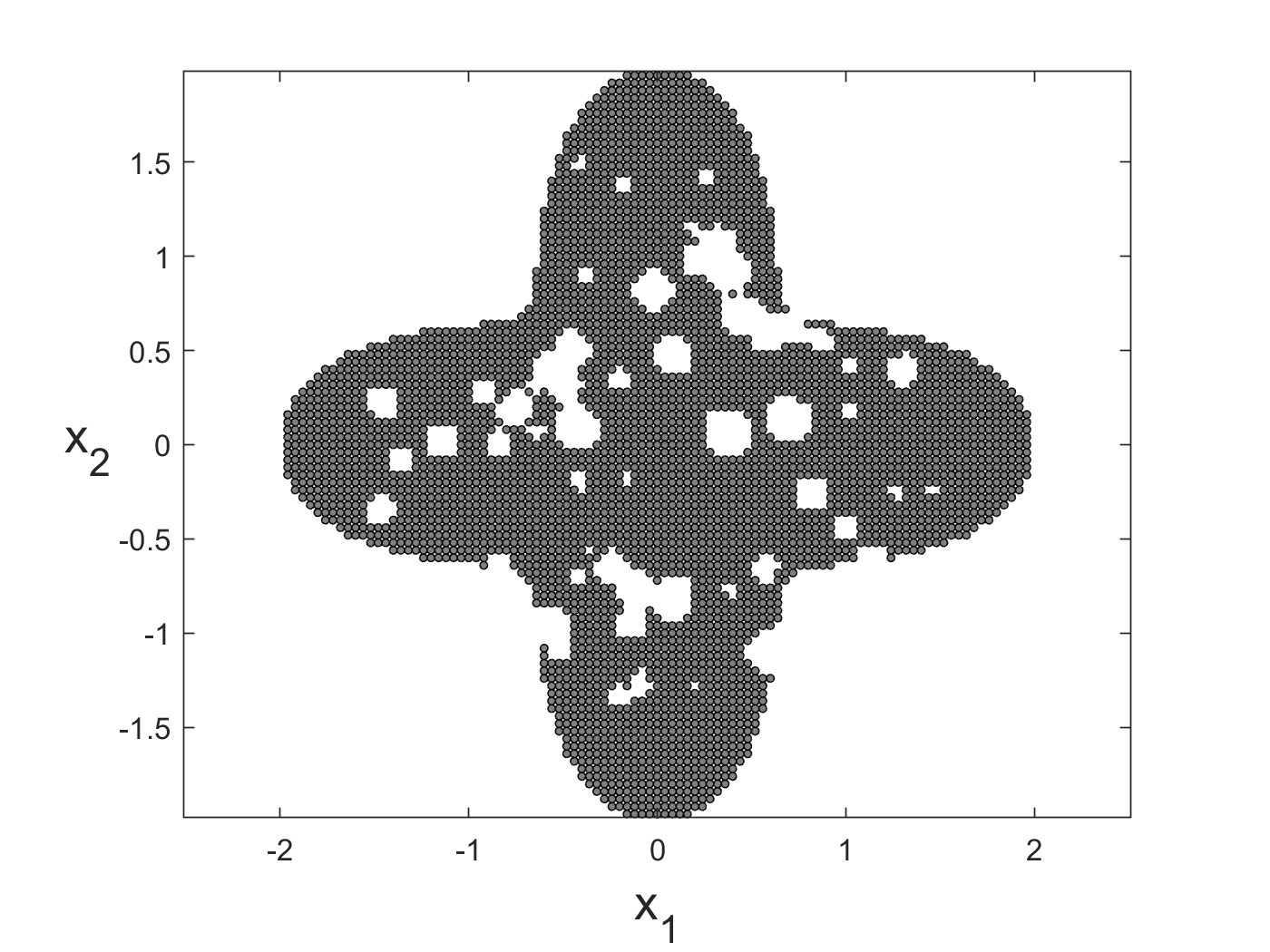}
		\caption{One quarter of the domain $\Omega$: it is a smooth trefoil filled with small dielectic rods. Holes have been created by randomly removing rods. \label{astroid}}
	\end{center}
\end{figure}
The domain $\Omega$ is not entirely filled with cylinders: "holes" have been created in $\Omega$ by removing cylinders.

For numerical purpose, the unit of length is defined by the wavelength and we choose $\lambda=1$. The domain $\Omega$ contains $N=4657$ rods with relative permittivity 12. We use the multiple scattering approach described in \cite{moijosa}. 

We assume that the rods, at positions $(\x_k)$, are small enough that they can each be characterized by only one scattering coefficient $s^0_q$ \cite{moijosa}. We denote $\x=(x_1,x_2)$.
\begin{figure}[h!]
	\begin{center}
		\includegraphics[width=10cm]{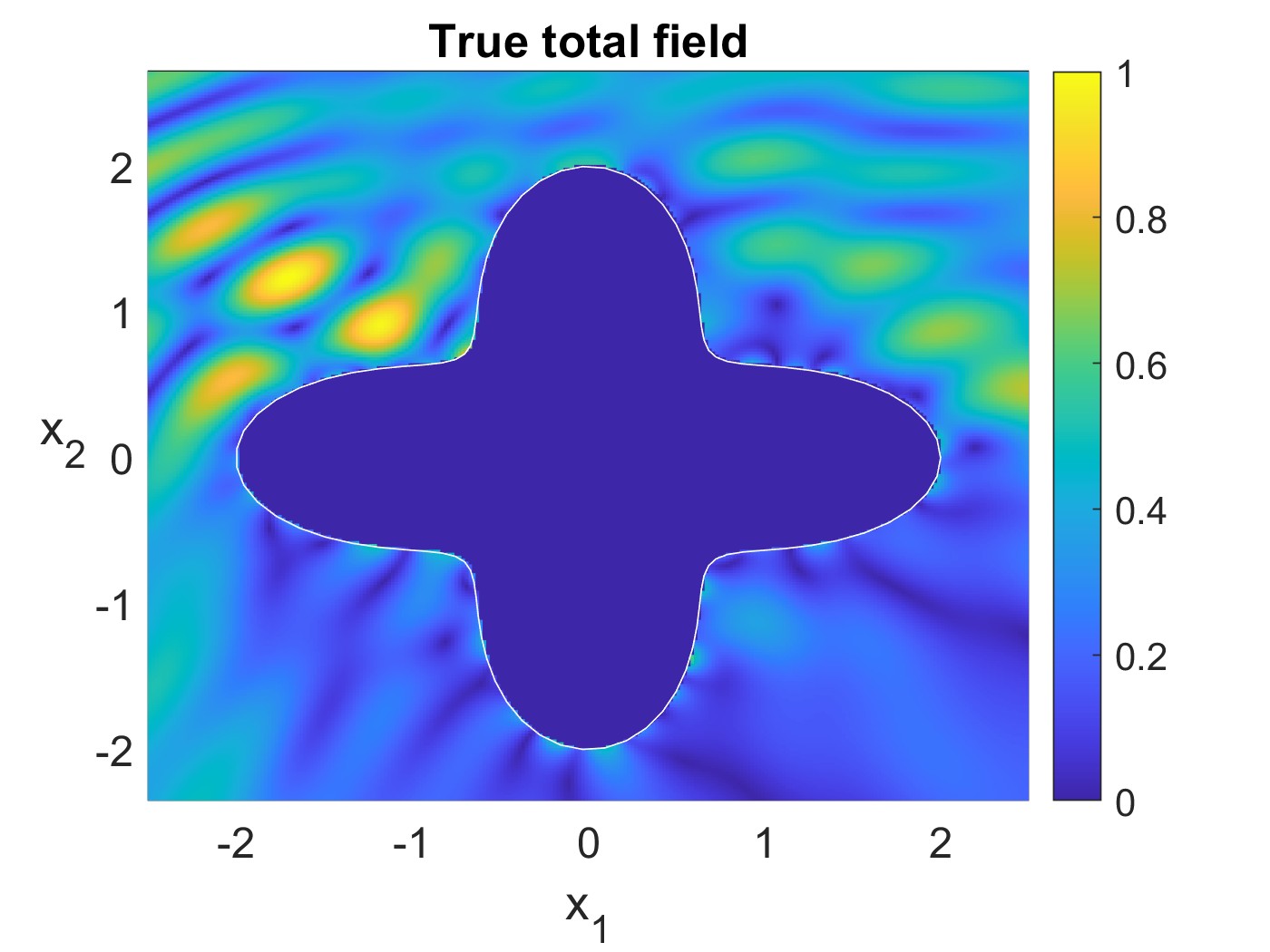}
		\caption{Map of the normalized modulus of the total field outside the trefoil, computed by a direct summation over all the scatterers contained inside $\Omega$.\label{true}}
	\end{center}
	\begin{center}
		\includegraphics[width=10cm]{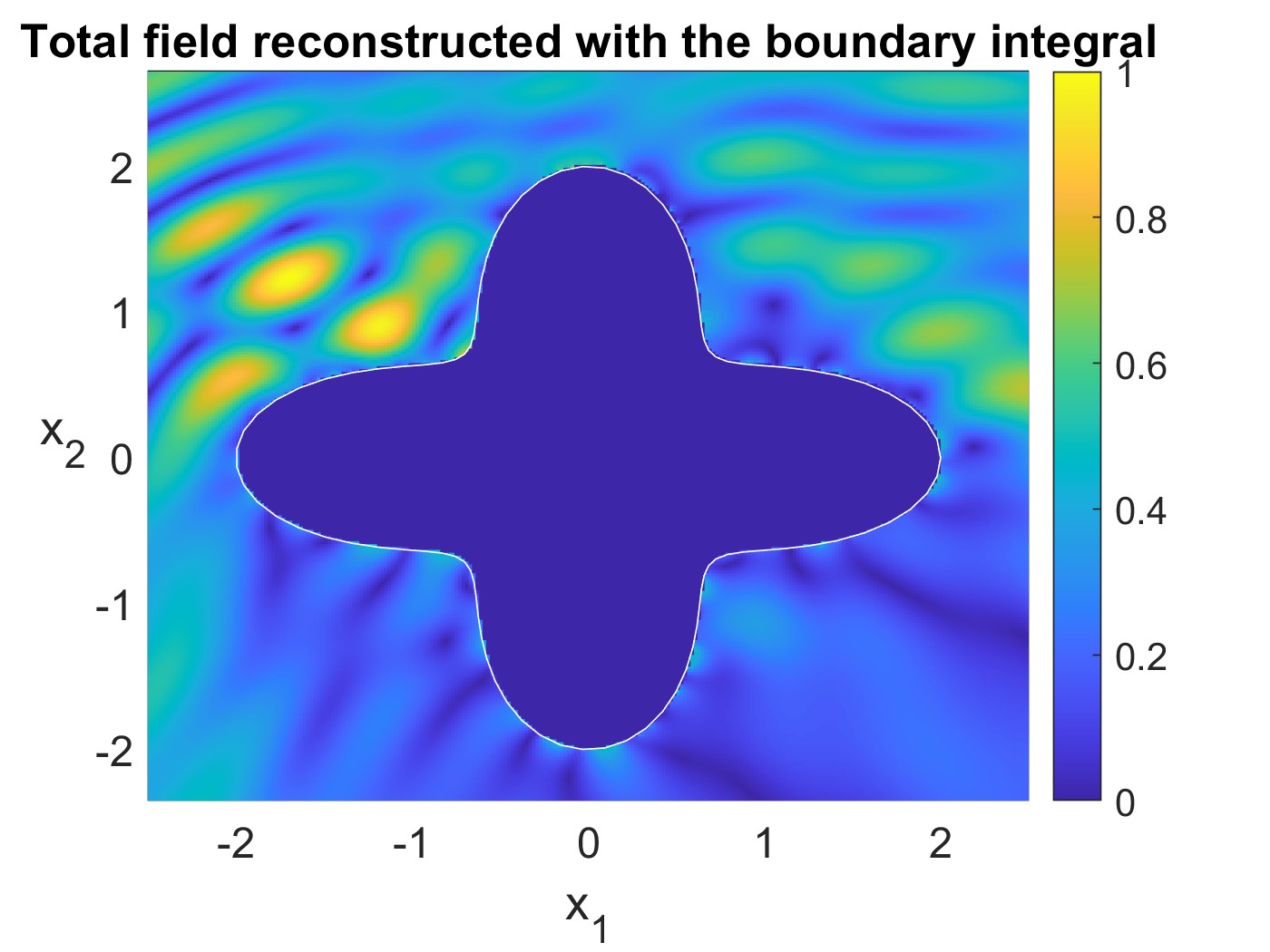}
		\caption{Map of the normalized modulus of the total field outside the trefoil, computed by using the single layer representation of the scattered field. \label{reco}}
	\end{center}
\end{figure}
At this step the scattered field is represented, everywhere outside the cylinders, by a sum over the rods in the form: 
\bq \label{MST}
u^s(\x)=\sum_{q=1}^N s^0_q  H^{(1)}_0(k|\x-\x_q|).\eq
 The coefficients $\hs=(s^0_q)_{1,\hdots, N}$ are determined from the multiple scattering theory exposed in \cite{moijosa}, see also \cite{foldy,lax} for historical references on this approach.
The scattering coefficient $s^{0}_q$ is related to the local incident field $u^{\rm inc,loc}_q$ through the scattering amplitude $t^{0}_q$: $s^{0}_q=t^{0}_q u^{\rm inc,loc}_q$. For a circular dielectric rod of radius $r_q$ and relative permittivity $\varepsilon_q=\nu_q^2$, it holds: 
$$
t^0_q=-\frac{J_1(kr_q)J_0(k\nu_q r_q)-\nu_q J_0(kr_q) J_1(k\nu_q r_q)}
    {H^{(1)}_1(kr_q) J_0(k\nu_q r_q)-\nu_q H^{(1)}_0(kr_q) J_1(k\nu_q r_q)}.
$$
 The local incident field is:
\bq
u^{\rm inc,loc}_q=u^{\rm inc}(\x_q)+\sum_{j \neq q} s^0_j H_0^{(1)}(|\x_q-\x_j|).
\eq
Therefore, it holds:
\bq
s^0_q=t^0_q \left(u^{\rm inc}(\x_q)+ \sum_{j \neq q} s^0_j H_0^{(1)}(|\x_q-\x_j|)\right).
\eq
Let us denote $\dbar{t}$ the diagonal matrix defined by $\dbar{t}=diag(t^0_1,\ldots,t^0_N)$ and $\dbar{h}$ the matrix with entries $h_{ij}=H_0^{(1)}(|\x_i-\x_j|)$ for $i \neq j$ and $h_{ii}=0$. The coefficients $\hs$ are obtained by solving the system:
\bq
\left( (\dbar{t})^{-1}-\dbar{h} \right) \hs=\hu^{\rm inc}
\eq
where $\hu^{\rm inc}=(u^{\rm inc}(\x_q))_{1,\hdots, N}$.
\begin{figure}[h!]
	\begin{center}
		\includegraphics[width=10cm]{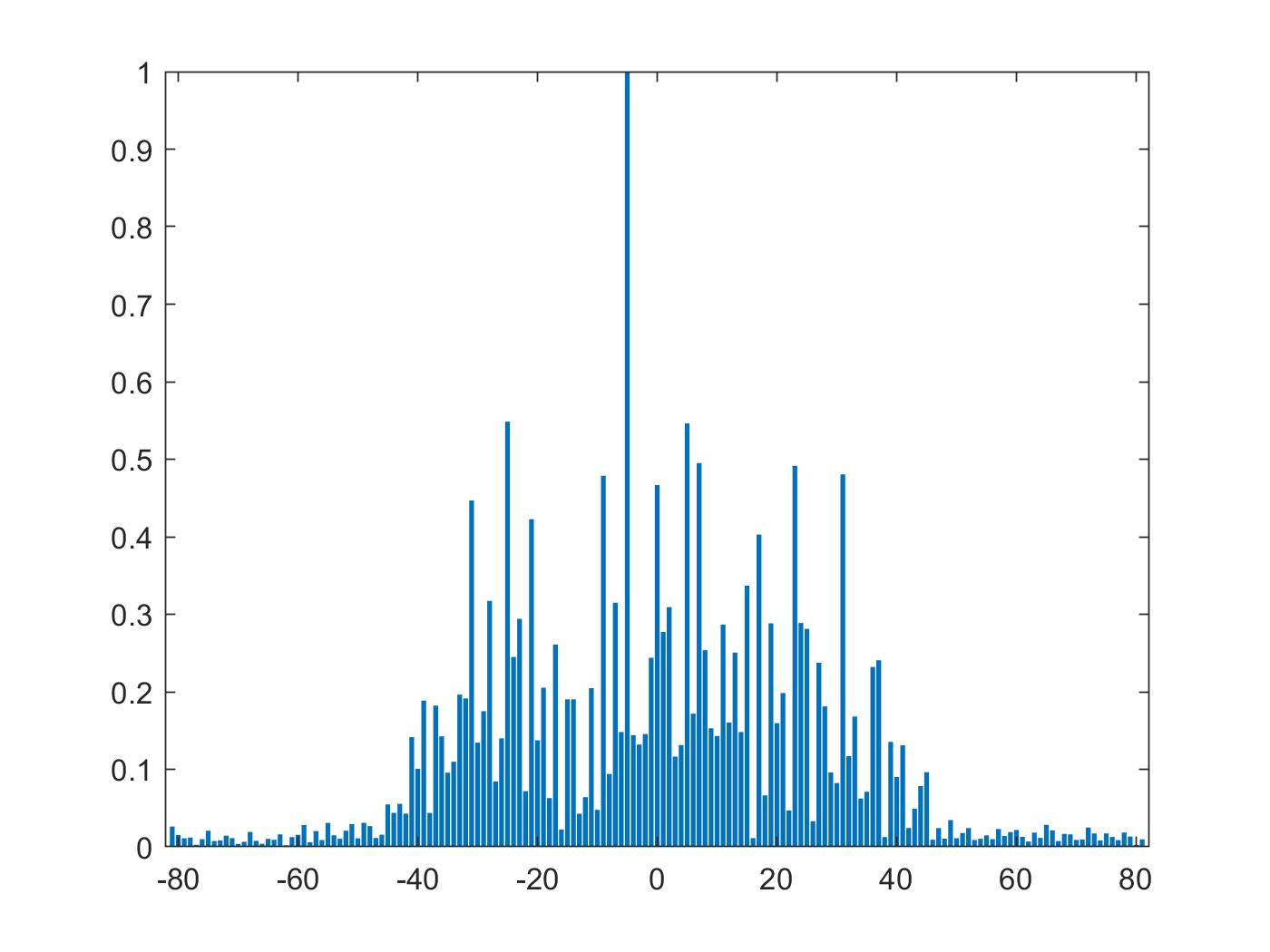}
		\caption{Normalized Discrete Fourier Transform coefficients of the sequence of coefficients $(s_k)$.\label{dft}}
	\end{center}
\end{figure}

\begin{figure}[h!]
	\begin{center}
		\includegraphics[width=10cm]{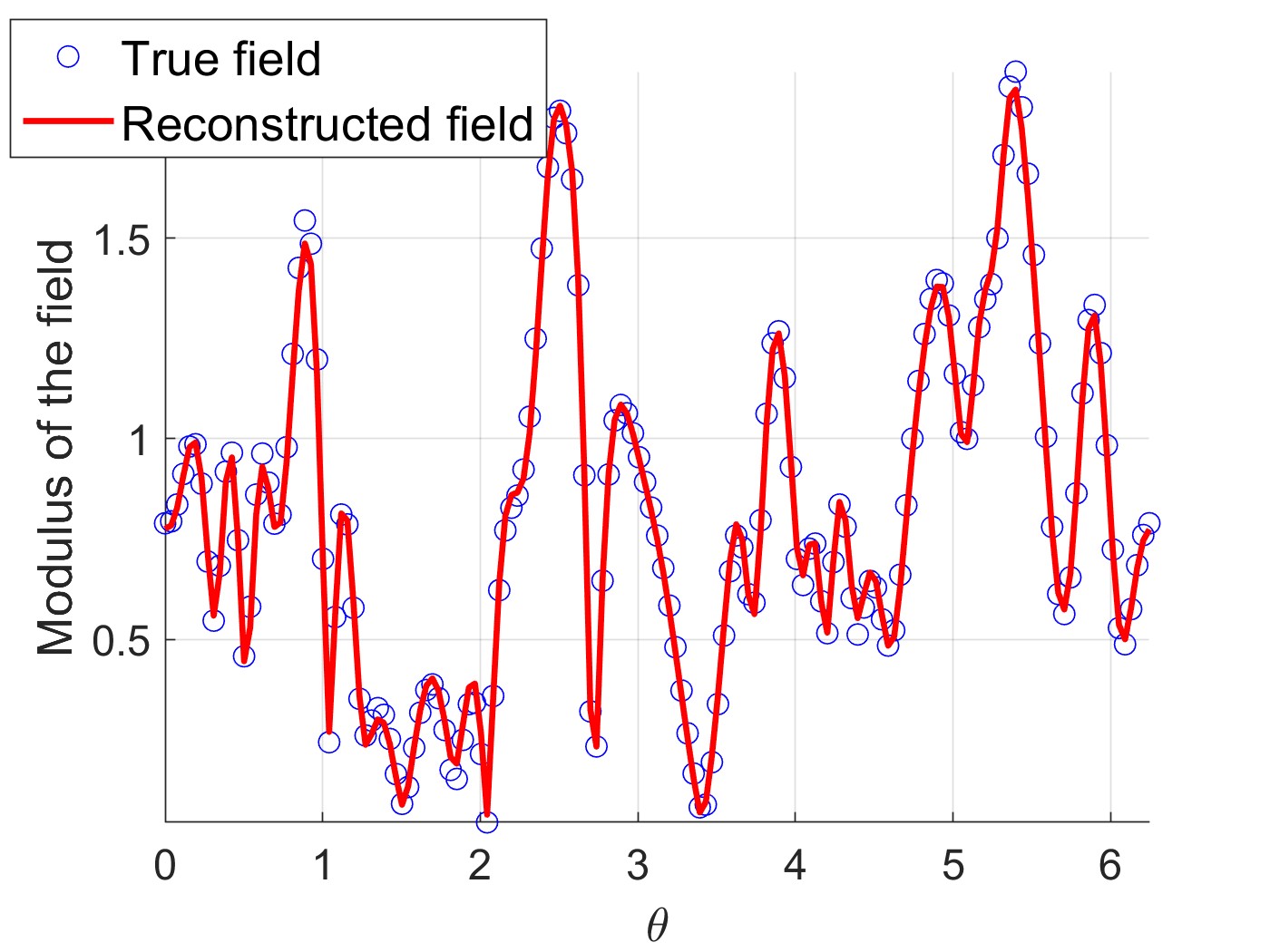}
		\caption{Modulus of the total field on a curve deduced from $\Gamma$ by a homothety of ratio 1.3. The red curve corresponds to the total field reconstructed by means of the single layer representation and the blue circles correspond to the total field computed by a direct summation over all the scatterers.\label{homo}}
	\end{center}
\end{figure}
\textcolor{black}{
The point is now to be able to represent the scattered field by means of a single layer potential as explained in section (\ref{intrep}). The existence of $\sigma^s$ being ensured by proposition (\ref{scatint}), it can be determined by solving the following integral equation:
$$
\int_{\Gamma} \sigma^s(\y') H^{(1)}_0(k|\y-\y'|) \, ds(\y')=u^s(\y),\y\in \Gamma,
$$
so as to obtain an approximation of the operator $\calT_{\rm loc}$ described in corollary (\ref{fundlemma}).
In order to do so numerically, that is, to obtain a discretized version of the density $\sigma^s$, we write a discrete version of the integral. Denoting $I_{np}=H^{(1)}_0(k|\y'_n-\y_p|)$, we have: 
\bq \label{sgl}
u^s(\y'_m)=\sum_{p=1}^{P} \sigma^s_p I_{np}, \,m=1 \hdots M,
\eq
where the points $\{\y'_m\}$ are put uniformly on $\Gamma$. The $P$ points $\{\y_p\}$ are chosen amongst the set $\{\frac{1}{2}(\y'_{p+1}+\y'_p)\}$, $\sigma_p^s$ is the average value of $\sigma^s$ over the arc $\wideparen{\y_p\y_{p+1}}$, times $\ell(\wideparen{\y_y'\y_{p+1}})$. 
 Finally, an overdetermined linear system ($M > P$) is obtained.
In matrix form, this can be written:
\bq
\dbar{I} \hsig^s=\hu^s,
\eq
where: $\hsig^s=(\sigma^s_p),\,\hu^s=(u^s(\x_m)) $ and $\dbar{I}$ is a $M \times P$ matrix with entries $I_{np}$. It is solved by means of a least square algorithm. On $\Gamma$, the values $(u^s(\y'_m))$ of the scattered field are otained by summing over the scatterers:
$$
 u^s_m \equiv u^s(\y'_m)=\sum_{q=1}^N s^0_q H^{(1)}_0(k|\y'_m-\x_q|), \, \y'_m \in \Gamma, m=1\ldots M.
$$
In matrix form, this reads as:
\bq
\dbar{H}  \hs=\hu^s.
\eq
Let us remark that $\dbar{H}$ is a $M \times N$ matrix. This system is solved by means of an iterative algorithm.
}
%We have now to handle the following integrals:
%$$
%I_{np}=\int_{\wideparen{\y_p\y_{p+1}}} H^{(1)}_0(k|\y'_n-\y|) d\y.
%$$
%For $n\neq p$, there is no singularity for $\y'_p \notin \wideparen{\y_p\y_{p+1}}$ and we can approximate $I_{np}$ by:
%$$
%I_{np}=H^{(1)}_0(k|\y'_n-\y_p|) \, \ell(\wideparen{\y_p\y_{p+1}}) \,,
%$$
%where $\ell(\wideparen{\y_p\y_{p+1}})$ is the length of the arc $\wideparen{\y_p\y_{p+1}}$. For $n=p$, there is an integrable singularity which can be handled in the folling manner, without entering into too much details. Let us note that $H_0^{(1)}(|x|) = 1+\frac{2i}{\pi} (\log(|x|)+\gamma-\log(2))+O(x^2)$. Then choose a parametrization $s \to \gamma(s)$ of $\Gamma$ and define $(s_p)$ and $(s'_p)$ such that $\y_p=\gamma(s_p)$ and $\y'_p=\gamma(s'_p)$. 
%Denoting: $K_0(s,s'_n)=H_0^{(1)}(k |\gamma(s)-\gamma(s'_n)|) \gamma'(s)$
%Then we obtain for some $\delta>0$:
%\begin{eqnarray*}
%I_{nn}=\int_{s_n}^{s_{n+1}} K_0(s,s'_n)ds =  \\ 
%\int_{s_n}^{s'_n-\delta/2} K_0(s,s'_n)ds+
%\int_{s'_n-\delta/2}^{s'_n+\delta/2} K_0(s,s'_n)ds+
%\int_{s'_n+\delta/2}^{s_{n+1}}K_0(s,s'_n)ds.
%\end{eqnarray*}
%Only the middle integral is singular, it is computed explicitely by using the series expansion given above.
% Finally, we obtain a matrix $\dbar{B}$ relating $\hs^s$ to $\hsig^s$:
%\bq \label{matB}
%\dbar{B}=(\dbar{I})^{-1} \dbar{H}.
%\eq
%This matrix is a $P \times N$ matrix.
Heuristically, the number $P$ can be determined by recalling that the function $\sigma^s$ is periodic, since it is defined on a bounded curve. Consequently, the computation of the DFT of $(\sigma^s_p)_{p \in \{1,\ldots P\}}$ can indicate whether the approximation is good, by checking the decreasing of the Fourier coefficients. This is exemplified in Figure~\ref{dft} where we have computed the DFT of the finite sequence $(\sigma^s_p)$. It is important to have this criterion, since the discrete values of the density $(\sigma^s_p)$ do not have a decreasing behavior with $P$. The final value is $P=160$.
%This is a rather brute force approach but our point is not to described an efficient numerical method, but rather to illustrate our theoretical results. A more refined method will be proposed in a forthcoming work. The main problem in computing the field in the very close vicinity of the object is the singular behavior of the Green function. Using a modal expansion replaces the singularity by an oscillating term which can be handled easily \cite{demanet}. This can be accomplished through the representation of the single layer operator by the symbol of a pseudodifferential operator \cite{boundary integrals Hsiao} .

We are able to reconstruct with a very good precision the diffracted field. In fig.~\ref{true}, we have plotted a map of the total field outside the region where the scatterers are contained, obtained by summing the contributions of the dielectric rods, and in fig.~\ref{reco}, it is the reconstructed field. Both fields have been normalized so that their maximal value is equal to 1, in order to have the same color scale. For a more direct comparison, in fig.~\ref{homo}, we have plotted the total field on a curve deduced from $\Gamma$ by a homothety of ratio 1.3.  We stress that, thanks to this approach the representation of the scattered field is now ensured by $160$ terms instead of 4657.
\subsection{Numerical example for the Fast Monopoles Method}
\textcolor{black}{
 We consider now the situation depicted in fig. \ref{astroids5}, consisting of 5 rotated copies of the trefoil with the difference that not all the rods were kept. As before, each of the astroids is filled with dielectric rods but this time the radius is $0.02$ \footnote{For a radius of $0.01$, the direct multiple scattering approach requires more than $80$Gb of memory, exceeding the available memory on our system. However the Fast Monopole Method still works.}. The relative permittivity is $12$. There are in total $9890$ rods. The trefoils are disposed in such a way that they cannot be enclosed into non overlapping disk, so that the expansion in cylindrical harmonics cannot be used. The curves $\Gamma_j$ supporting the single layer densities are indicated.
\begin{figure}[h!]
	\begin{center}
		\includegraphics[width=10cm]{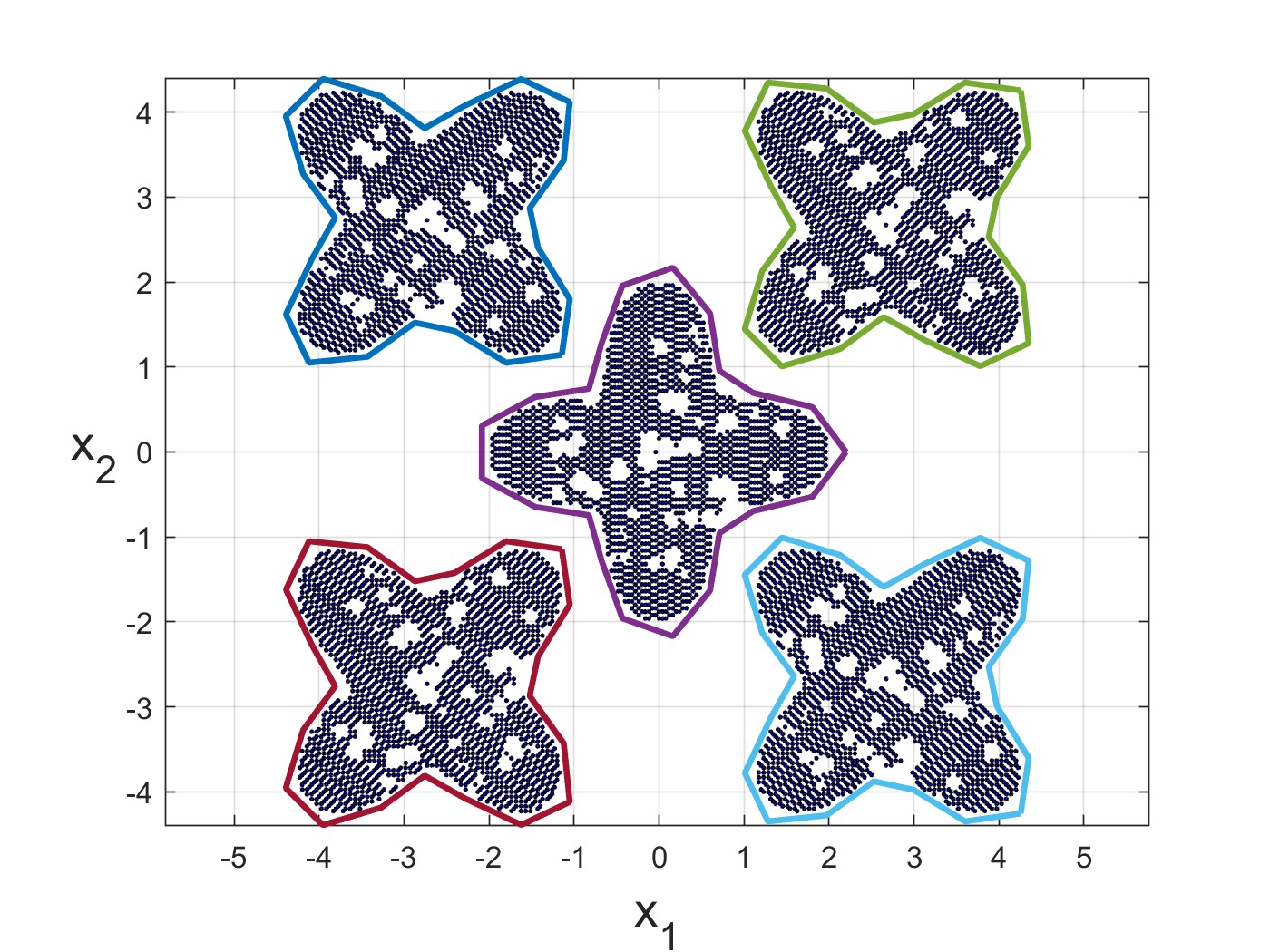}
		\caption{A collection of $9890$ scatterers contained in 5 domains. The curves around the domains are those supporting the monopoles. They are discretized by means of the number $M$ as given in table \ref{tablevaleur}.\label{astroids5}}
	\end{center}
\end{figure}
  The incident field is a plane wave $u^{\rm inc}(\x)=e^{-ik x_2}$. As in the first numerical example, the unit of length is that of the wavelength and we choose three values of the wavelength: $\lambda=20,10,5$.
%    The maps of the field is given in Figure~\ref{maps}. We have computed the maps by using the single layer representation (left panel) and by using the multiple scattering theory for the entire set of rods in the right panel. We have used $P=22,$ points to compute the single layer representations (\ref{sgllayer}). The maps and the scattering coefficients agree to a precision below $0.7\%$ (in $L^2$ norm for the entire region covered by the map).
%\begin{figure}[h!]
%	\begin{center}
%		\includegraphics[width=10cm]{maps}
%		\caption{Maps of the modulus field. On the left (a) panel, the field is computed by using the extended Fast Multipole Method, on the right (b) panel, it is computed by using directly the multiple scattering theory for the entire set of cylinders. \label{maps}}
%	\end{center}
%\end{figure}
In fig. \ref{lambda20}, we have plotted the modulus of the scattered field on a circle of radius $8$.  The fields coincide to better than $1.5\%$. 
\begin{figure}[h!]
	\begin{center}
		\includegraphics[width=10cm]{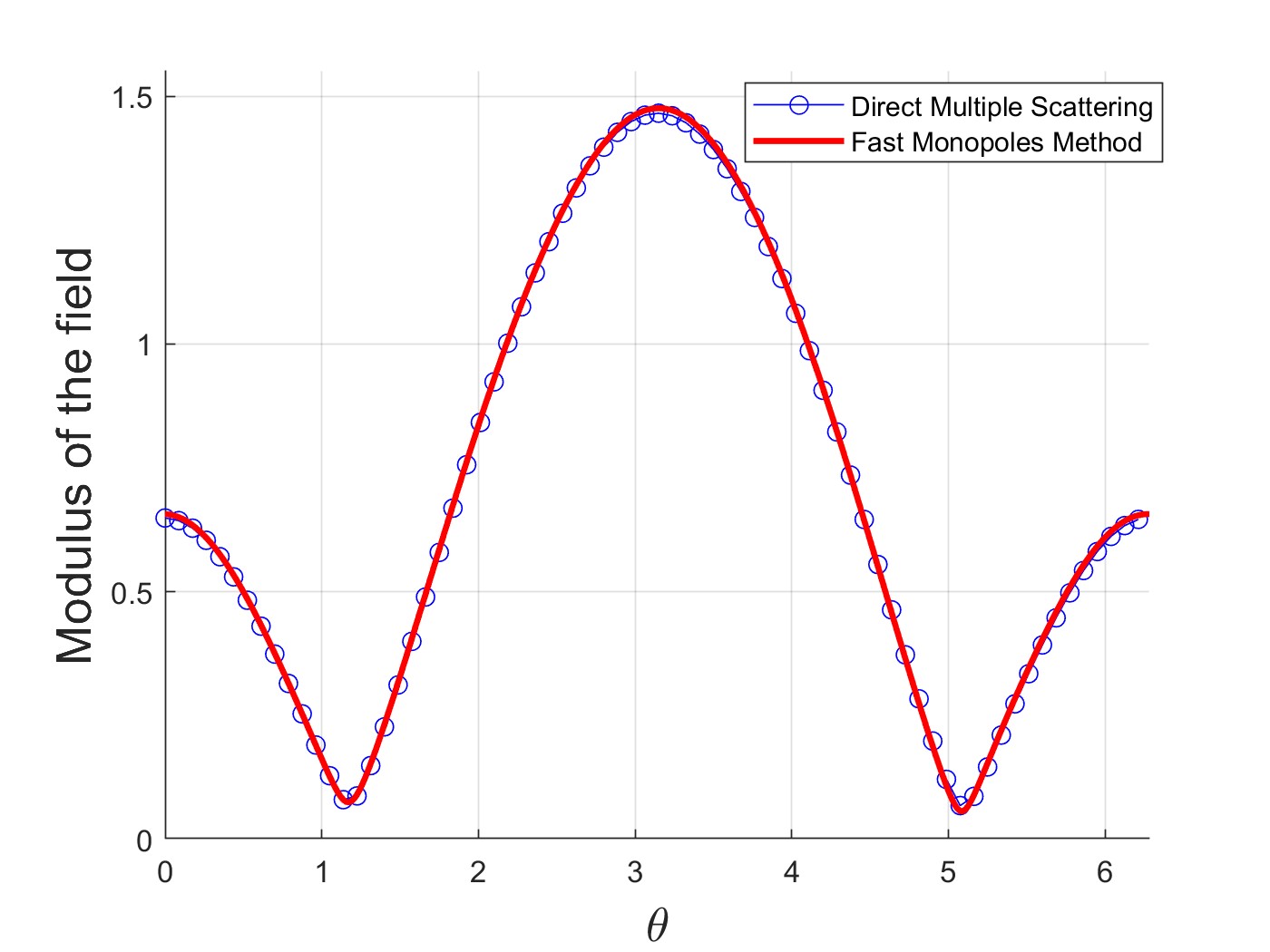}
		\caption{The field is computed on a circle of radius $8$ enclosing the scatterers. Here $\lambda=20$. The curve in solid line corresponds to the field computed by using directly the single layer representations of (cf. (\ref{sgl})). The  circles correspond to the field computed by means of the direct multiple scattering approach (cf. (\ref{MST})).  \label{lambda20}}
	\end{center}
\end{figure}
The running time for the brute force method and the Fast Monopole Method is given in table \ref{tablevaleur}. As the wavelength gets smaller, the sampling of the $\Gamma$ curves should be finer and the gain is reduced. 
\begin{table}[h!]
	\begin{center}
		\caption{Computation times, number of monopoles $P$ and number of points $M$ for the computation of the density. FMM stands for  the Fast Monopole Method and MS for the direct multiple scattering approach. The number of monopoles was chosen so as to obtain a error of less than $1.5$\% as compared to the direct method. The field is computed on a circle of radius $8$ enclosing all the scatterers.}
		\label{tablevaleur}
		\begin{tabular}{l|c|c|r} 
			\textbf{$\lambda$}& P (M) & \textbf{Time (FMM)} & \textbf{Time (MS)}\\
			%$\alpha$ & $\beta$ & $\gamma$ \\
			\hline
			20 & 7 (66) & 45s & 1300s\\
			10 & 14 (140) & 90s & 1400s\\
			5 & 27 (265) & 220s & 1500s\\
		\end{tabular}
	\end{center}
\end{table}
\begin{figure}[h!]
	\begin{center}
		\includegraphics[width=10cm]{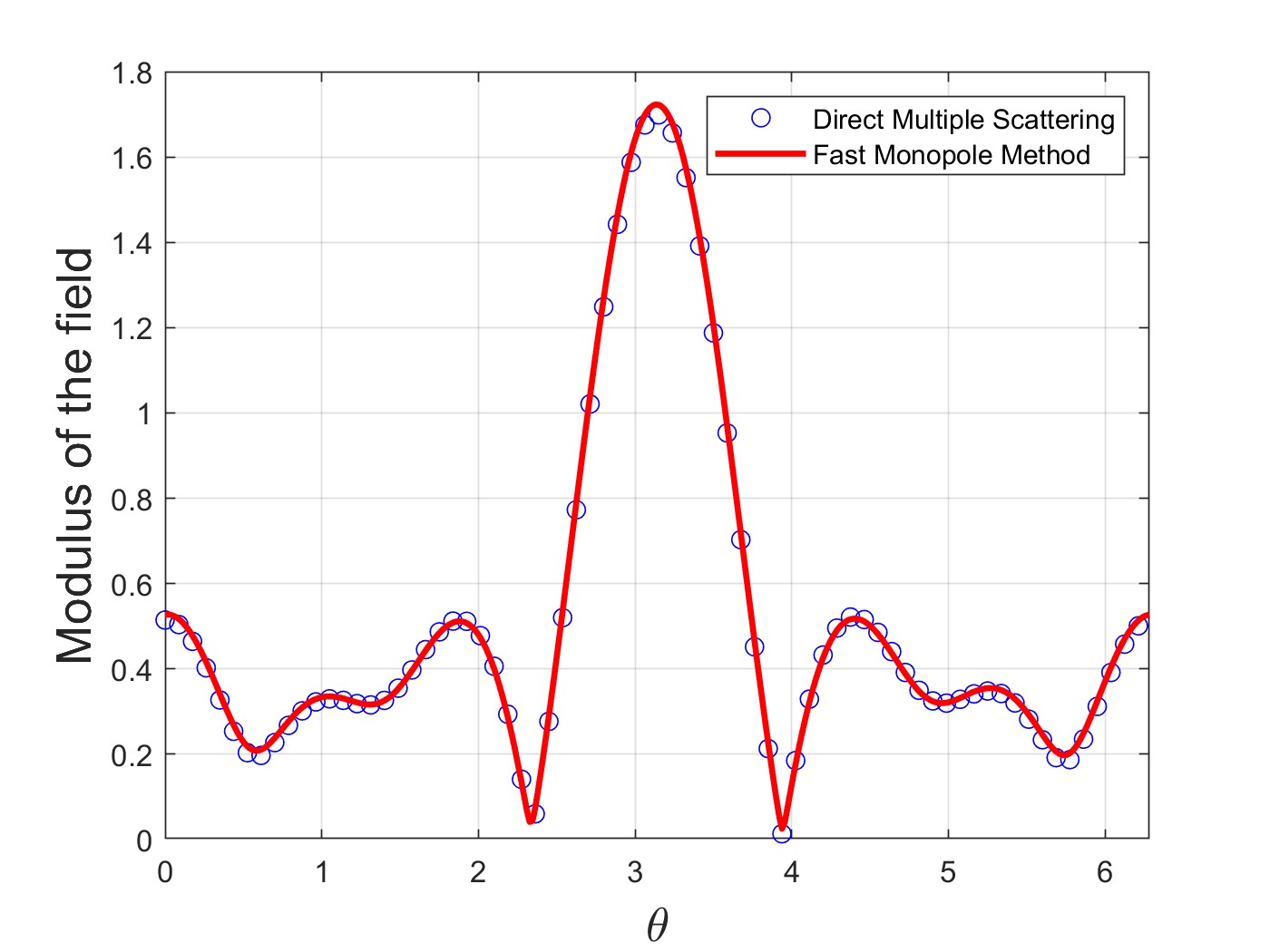}
		\caption{Same as fig. (\ref{lambda20}) with $\lambda=10$. \label{lambda10}}
	\end{center}
\end{figure}
\begin{figure}[h!]
	\begin{center}
		\includegraphics[width=10cm]{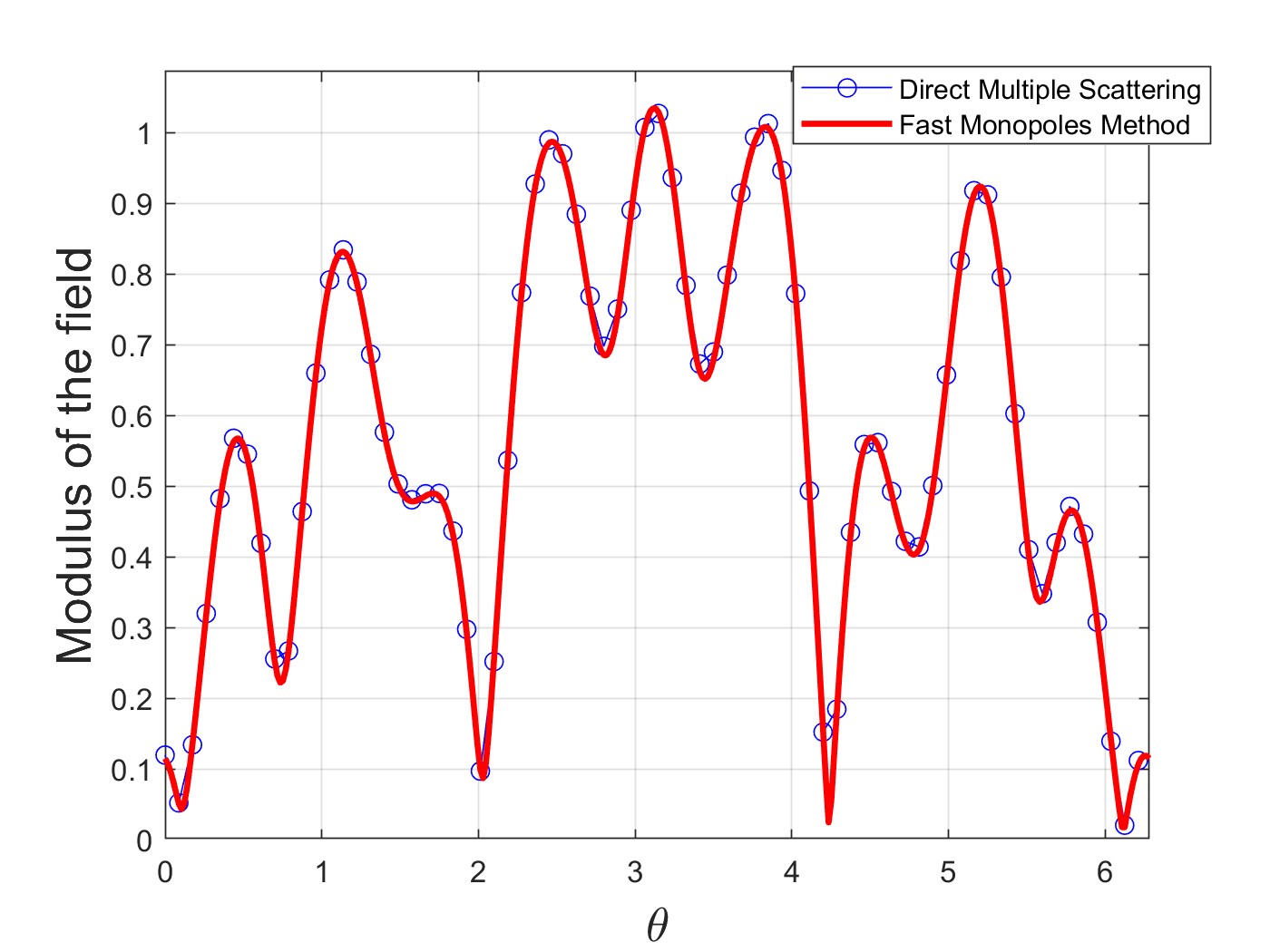}
		\caption{Same as fig. (\ref{lambda20}) with $\lambda=5$. \label{lambda5}}
	\end{center}
\end{figure}
}
In conclusion, we have proposed a new way of representing the field scattered by a collection of objects, by using a single layer representation. The scattered field is characterized by a density supported by any smooth surface enclosing the scatterers. From a numerical point of view, this representation allows to represent the field with a much smaller number of parameters, as compared to the direct representation as a sum over the scatterers. Since the monopoles are supported by a region of codimension 1, much less information is needed, as compared to a volumetric representation (of codimension 0). This result in a drastic reduction of the number of values required for representing the scattered field with a given precision. This results in a generalization of the representation of the field by spherical harmonics used in the Fast Multipole Method and extends this algorithm beyond the spherical geometry.

\end{document}